\documentclass[aps,pra]{revtex4}
\usepackage{amsmath,amssymb,amsthm,graphicx,bbm}
\usepackage[
	bookmarks = false,
	colorlinks = true,
	linkcolor = blue,
	urlcolor= black,
	citecolor = blue]{hyperref}
\newtheorem{theorem}{Theorem}
\newtheorem{lemma}{Lemma}
\newtheorem{corollary}{Corollary}

\newtheorem{proposition}{Proposition}

\begin{document}
\title{Bound on optimal local discrimination of multipartite quantum states}
\author{Donghoon Ha}
\affiliation{Department of Applied Mathematics and Institute of Natural Sciences, Kyung Hee University, Yongin 17104, Republic of Korea}
\author{Jeong San Kim}
\email{freddie1@khu.ac.kr}
\affiliation{Department of Applied Mathematics and Institute of Natural Sciences, Kyung Hee University, Yongin 17104, Republic of Korea}

\begin{abstract}
We consider the unambiguous discrimination of multipartite quantum states and provide an upper bound for the maximum success probability of optimal local discrimination. We also provide a necessary and sufficient condition to realize the upper bound. We further establish a necessary and sufficient condition for this upper bound to be saturated.  Finally, we illustrate our results using examples in multidimensional multipartite quantum systems.
\end{abstract}
\maketitle

\section*{INTRODUCTION}
\indent Quantum nonlocality is a quintessential phenomenon of multipartite quantum systems which does not have any classical counterpart. Entanglement is one of the most representative nonlocal quantum correlations which cannot be realized only by \emph{local operations and classical communication} (LOCC)\cite{chit20142,horo2009}. It is known that the nonlocal property of quantum entanglement can be used as a resource in many quantum information processing tasks\cite{chit2019}.\\
\indent Quantum nonlocal phenomenon can also arise in multipartite quantum state discrimination, which is an essential process for efficient information transfer in quantum communication. In general, orthogonal quantum states can be discriminated with certainty, whereas such discrimination is impossible for nonorthogonal quantum states. Along this line, state discrimination strategies are needed to discriminate nonorthogonal quantum states at least with some nonzero probability\cite{chef2000,barn20091,berg2010,bae2015}. However, some orthogonal states of multipartite quantum systems cannot be discriminated with certainty when the available measurements are limited to LOCC measurements\cite{benn19991}. As orthogonal states can always be discriminated with certainty when there is no limitation of possible measurement, this limited discrimination ability of LOCC measurement reveals the nonlocal phenomenon inherent in quantum state discrimination.\\
\indent Nonlocal phenomenon of quantum state discrimination can also arise in discriminating nonorthogonal states of multipartite quantum systems; it is known that some nonorthogonal states cannot be optimally discriminated using only LOCC\cite{pere1991,duan2007,chit2013}. For this reason, much attention has been
shown for the optimal local discrimination of multipartite quantum states \cite{ghos2001,walg2002,fan2004,duan2009,chit20141,band2015,band2021,zhan2020}. Nevertheless, realizing optimal local discrimination still remains a challenging task because it is hard to have a good mathematical characterization of LOCC.\\
\indent One efficient way to overcome this difficulty is to investigate possible upper bounds for the maximum success probability of optimal local discrimination. 
For a better understanding of optimal local discrimination, it is also important to establish good conditions realizing such upper bounds. Recently, an upper bound of maximum success probability was established in local minimum-error discrimination of bipartite quantum states. Moreover, a necessary and sufficient condition was also provided for this upper bound to be saturated\cite{ha2022}.\\
\indent Here, we consider \emph{unambiguous discrimination} (UD)\cite{ivan1987,pere1988,diek1988,zhan2022} among multipartite quantum states of arbitrary dimensions and provide an upper bound for the maximum success probability of optimal local discrimination. Moreover, we provide a necessary and sufficient condition to realize this upper bound. We also establish a necessary and sufficient condition for this upper bound to be saturated. Finally, we illustrate our results using examples in multidimensional multipartite quantum systems.\\
\indent This paper is organized as follows. In the ``\hyperref[ressec]{Results}'' Section, we first recall the definition and some properties about separable operators and separable measurements in multipartite quantum systems. We further recall the definition of UD and provide some useful properties of optimal UD (Proposition~\ref{pro:nsc}).
As the main results of this paper, we provide an upper bound for the maximum success probability of optimal local discrimination by using a certain class of Hermitian operators acting on multipartite Hilbert space (Theorem~\ref{thm:qupb}). Moreover, we provide a necessary and sufficient condition for the Hermitian operator to realize this upper bound (Theorem~\ref{thm:nsc} and Corollary~\ref{cor:qupb}). We also establish a necessary and sufficient condition for this upper bound to be saturated (Corollary~\ref{cor:plqg}). We illustrate our results by examples in multidimensional multipartite quantum systems (Examples~1 and 2). In the ``\hyperref[mtdsec]{Methods}'' Section, we provide a detail proof of Theorem~\ref{thm:qupb}. In the ``\hyperref[dissec]{Discussion}'' Section, we summarize our results and discuss possible future works related to our results.

\section*{RESULTS}\label{ressec}
\indent In multipartite quantum systems, 
a state is described by a density operator $\rho$ 
acting on a multipartite Hilbert space $\mathcal{H}=\bigotimes_{k=1}^{m}\mathcal{H}_{k}$ 
consisting of the subsystems $\mathcal{H}_{k}\cong \mathbb{C}^{d_{k}}$ with $m\geqslant2$ and positive integers $d_{k}$ for $k=1,\ldots,m$. 
A measurement is expressed by a \emph{positive operator-valued measure} $\{M_{i}\}_{i}$ that is a set of positive-semidefinite operators $M_{i}\succeq0$ on $\mathcal{H}$ satisfying completeness relation $\sum_{i}M_{i}=\mathbbm{1}$, where $\mathbbm{1}$ is the identity operator on $\mathcal{H}$. The probability of obtaining the measurement outcome corresponding to $M_{j}$ is $\mathrm{Tr}(\rho M_{j})$ for the prepared state $\rho$.\\
\indent A positive-semidefinite operator $E$ (not necessarily a state) on $\mathcal{H}$ is called \emph{separable} if it can be represented as a summation of
positive-semidefinite product operators, that is,
\begin{equation}
E=\sum_{l}\bigotimes_{k=1}^{m}E_{l,k},
\end{equation}
where $E_{l,k}$ is a positive-semidefinite operator on $\mathcal{H}_{k}$ for each $k=1,\ldots,m$.
We denote the set of all positive-semidefinite \emph{separable} operators on $\mathcal{H}$ as
\begin{equation}\label{eq:sep}
\mathrm{SEP}=\big\{
E\,\big|\,
\mbox{$E$: a positive-semidefinite separable operator acting on $\mathcal{H}$}\big\}.
\end{equation}
A measurement $\{M_{i}\}_{i}$ is called a \emph{LOCC measurement} if it can be implemented by LOCC, and 
it is called \emph{separable} if $M_{i}\in\mathrm{SEP}$ for all $i$. 
It is known that every LOCC measurement is separable\cite{chit20142}.\\
\indent A simple example of LOCC measurement is  
$\{M_{1,i_{1}}\otimes\cdots\otimes M_{m,i_{m}}\}_{i_{1},\ldots,i_{m}}$
where $\{M_{l,i_{l}}\}_{i_{l}}$ is a measurement on the $l$th subsystem for each $l=1,\ldots,m$.
This separable measurement can be implemented by local measurements and classical communication;
For each $l=1,\ldots,m$, a local measurement $\{M_{l,i_{l}}\}_{i_{l}}$ is performed on the $l$th subsystem;
If it is confirmed through classical communication that the measurement result of the $l$th subsystem is $M_{l,i_{l}}$ for each $l=1,\ldots,m$, the measurement result of the whole system becomes $M_{1,i_{1}}\otimes\cdots\otimes M_{m,i_{m}}$. Thus, $\{M_{1,i_{1}}\otimes\cdots\otimes M_{m,i_{m}}\}_{i_{1},\ldots,i_{m}}$ is a LOCC measurement.\\ 
\indent Here, we consider the situation of discriminating \emph{multipartite} quantum states $\rho_{1},\ldots,\rho_{n}$ where the state $\rho_{i}$ is prepared with the probability $\eta_{i}$. We denote this situation as an ensemble $\mathcal{E}=\{\eta_{i},\rho_{i}\}_{i=1}^{n}$.
Moreover, we consider the discrimination of the multipartite state ensemble $\mathcal{E}$ using a measurement $\{M_{i}\}_{i=0}^{n}$ where $M_{0}$ gives inconclusive results about the prepared state and
$M_{i}$ provides conclusive results of $\rho_{i}$ for each $i=1,\ldots,n$.
For the conclusive results to be unambiguous, so-called \emph{no-error condition} is required;
\begin{equation}\label{eq:rmz}
\mathrm{Tr}(\rho_{i}M_{j})=0~~\forall i,j=1,\ldots,n~\mbox{with}~i\neq j.
\end{equation}
By defining
\begin{equation}\label{eq:posi}
\mathrm{Pos}_{i}(\mathcal{E})=
\{E\succeq0\,|\,\mathrm{Tr}(\rho_{j}E)=0~\forall j=1,\ldots,n~\text{with}~j\neq i\},
\end{equation}
Eq.~\eqref{eq:rmz} can be rewritten as
\begin{equation}\label{eq:nec}
M_{i}\in\mathrm{Pos}_{i}(\mathcal{E})~~\forall i=1,\ldots,n.
\end{equation}
We say that a measurement $\{M_{i}\}_{i=0}^{n}$ is \emph{unambiguous} if it satisfies the no-error condition in Eq.~\eqref{eq:nec}. \\
\indent The \emph{optimal UD} of $\mathcal{E}$ is to minimize the probability of obtaining inconclusive results.
Equivalently, the optimal UD of $\mathcal{E}$ is to achieve the optimal success probability of $\mathcal{E}=\{\eta_{i},\rho_{i}\}_{i=1}^{n}$ defined as
\begin{equation}\label{eq:msp}
p_{\rm G}(\mathcal{E})=\max_{\substack{\rm Measurement\\ \rm with~\mbox{\scriptsize\eqref{eq:nec}}}}\sum_{i=1}^{n}\eta_{i}\mathrm{Tr}(\rho_{i}M_{i}),
\end{equation}
where the maximum is taken over all possible unambiguous measurements.
The following proposition provide a necessary and sufficient condition
of an \emph{optimal} unambiguous measurement realizing $p_{\rm G}(\mathcal{E})$\cite{elda2004}.
\begin{proposition}\label{pro:nsc}
For given ensemble $\mathcal{E}=\{\eta_{i},\rho_{i}\}_{i=1}^{n}$,
an unambiguous measurement $\{M_{i}\}_{i=0}^{n}$ provides the optimal success probability $p_{\rm G}(\mathcal{E})$ if and only if there is a Hermitian operator $K$ satisfying the following condition,
\begin{subequations}\label{eq:nscud}
\begin{gather}
K\succeq0,\label{eq:nscuda}\\[1mm]
\mathrm{Tr}(M_{0}K)=0,\label{eq:nscudb}\\[1mm]
K-\eta_{i}\rho_{i}\in\mathrm{Pos}_{i}^{*}(\mathcal{E})~\forall i=1,\ldots,n,\label{eq:nscudc}\\[1mm]
\mathrm{Tr}[M_{i}(K-\eta_{i}\rho_{i})]=0~\forall i=1,\ldots,n,\label{eq:nscudd}
\end{gather}
\end{subequations}
where $\mathrm{Pos}_{i}^{*}(\mathcal{E})$($i=1,\ldots,n$) is the dual set of $\mathrm{Pos}_{i}(\mathcal{E})$ defined as
\begin{equation}\label{eq:posis}
\mathrm{Pos}_{i}^{*}(\mathcal{E})=\{
A\,|\,A^{\dagger}=A,~\mathrm{Tr}(AB)\geqslant0~\forall B\in\mathrm{Pos}_{i}(\mathcal{E})\}.
\end{equation}
In this case, we have
\begin{equation}
p_{\rm G}(\mathcal{E})=\sum_{i=1}^{n}\eta_{i}\mathrm{Tr}(\rho_{i}M_{i})=\mathrm{Tr}K.
\end{equation}
\end{proposition}

\indent For each $\mathrm{Pos}_{i}(\mathcal{E})$ in Eq.~\eqref{eq:posi},
we denote the subset of separable operators,
\begin{equation}\label{eq:sepi}
\mathrm{SEP}_{i}(\mathcal{E})=\mathrm{Pos}_{i}(\mathcal{E})\cap\mathrm{SEP}=
\{E\in\mathrm{SEP}\,|\,\mathrm{Tr}(\rho_{j}E)=0~\forall j=1,\ldots,n~\text{with}~j\neq i\}.
\end{equation}
The unambiguous measurement $\{M_{i}\}_{i=0}^{n}$ is called separable if 
\begin{equation}\label{eq:sepm}
M_{0}\in\mathrm{SEP},~M_{i}\in\mathrm{SEP}_{i}(\mathcal{E})~\forall i=1,\ldots,n,
\end{equation}
where $\mathrm{SEP}$ is defined in Eq.~\eqref{eq:sep}. \\
\indent When the available measurements are limited to separable unambiguous measurements, we denote the maximum success probability by
\begin{equation}\label{eq:psep}
p_{\rm SEP}(\mathcal{E})=\max_{\substack{\rm Measurement\\ \rm with~\mbox{\scriptsize\eqref{eq:sepm}}}}\sum_{i=1}^{n}\eta_{i}\mathrm{Tr}(\rho_{i}M_{i}).
\end{equation}
We use $p_{\rm L}(\mathcal{E})$ to denote
the maximum of success probability that can be obtained using LOCC unambiguous measurements, that is, 
\begin{equation}
p_{\rm L}(\mathcal{E})=\max_{\substack{\rm LOCC\\ \rm measurement\\ \rm with~\mbox{\scriptsize\eqref{eq:nec}}}}\sum_{i=1}^{n}\eta_{i}\mathrm{Tr}(\rho_{i}M_{i}).
\end{equation}
\indent For a given ensemble $\mathcal{E}=\{\eta_{i},\rho_{i}\}_{i=1}^{n}$, 
we define $q_{\rm SEP}(\mathcal{E})$
as the minimum quantity
\begin{equation}\label{eq:qsep}
q_{\rm SEP}(\mathcal{E})=\min\mathrm{Tr}H
\end{equation}
over all possible Hermitian operator $H$ satisfying
\begin{subequations}\label{eq:hcon}
\begin{gather}
H\in\mathrm{SEP}^{*},\label{eq:hcona}\\
H-\eta_{i}\rho_{i}\in\mathrm{SEP}_{i}^{*}(\mathcal{E})~\forall i=1,\ldots,n,\label{eq:hconb}
\end{gather}
\end{subequations}
where $\mathrm{SEP}^{*}$ and $\mathrm{SEP}_{i}^{*}(\mathcal{E})$($i=1,\ldots,n$) are the dual sets of $\mathrm{SEP}$ and $\mathrm{SEP}_{i}(\mathcal{E})$ in Eqs.~\eqref{eq:sep} and \eqref{eq:sepi}, respectively, that is,
\begin{equation}\label{eq:sepd}
\begin{array}{rcl}
\mathrm{SEP}^{*}&=&\{
A\,|\,A^{\dagger}=A,~\mathrm{Tr}(AB)\geqslant0~\forall B\in\mathrm{SEP}\},\\[1mm]
\mathrm{SEP}_{i}^{*}(\mathcal{E})&=&\{
A\,|\,A^{\dagger}=A,~\mathrm{Tr}(AB)\geqslant0~\forall B\in\mathrm{SEP}_{i}(\mathcal{E})\}.
\end{array}
\end{equation}
Note that $\mathrm{SEP}^{*}$ contains all positive-semidefinite operators because every element of $\mathrm{SEP}$ is positive semidefinite. Due to the similar reason,
$\mathrm{SEP}_{i}^{*}(\mathcal{E})$($i=1,\ldots,n$) contains all positive-semidefinite operators.
\\
\indent The following theorem shows that $q_{\rm SEP}(\mathcal{E})$ is equal to $p_{\rm SEP}(\mathcal{E})$ in Eq.~\eqref{eq:psep}.
The proof of Theorem~\ref{thm:qupb} is provided in the ``\hyperref[mtdsec]{Methods}'' Section.

\begin{theorem}\label{thm:qupb}
For a multipartite quantum state ensemble $\mathcal{E}=\{\eta_{i},\rho_{i}\}_{i=1}^{n}$, 
\begin{equation}\label{eq:ubsep}
p_{\rm SEP}(\mathcal{E})=q_{\rm SEP}(\mathcal{E}).
\end{equation}
\end{theorem}

\indent For a given ensemble $\mathcal{E}=\{\eta_{i},\rho_{i}\}_{i=1}^{n}$, the following theorem provides a necessary and sufficient condition on a Hermitian operator $H$ to realize $q_{\rm SEP}(\mathcal{E})$.

\begin{theorem}\label{thm:nsc}
For a multipartite quantum state ensemble $\mathcal{E}=\{\eta_{i},\rho_{i}\}_{i=1}^{n}$, a Hermitian operator $H$ satisfying Condition~\eqref{eq:hcon} gives $q_{\rm SEP}(\mathcal{E})$ if and only if there is a separable unambiguous measurement $\{M_{i}\}_{i=0}^{n}$ such that
\begin{subequations}\label{eq:nscqs}
\begin{gather}
\mathrm{Tr}(M_{0}H)=0,\label{eq:nscqsa}\\
\mathrm{Tr}[M_{i}(H-\eta_{i}\rho_{i})]=0~~\forall i=1,\ldots,n.\label{eq:nscqsb}
\end{gather}
\end{subequations}
In this case, we have
\begin{equation}
q_{\rm SEP}(\mathcal{E})=\mathrm{Tr}H=\sum_{i=1}^{n}\eta_{i}\mathrm{Tr}(\rho_{i}M_{i}).
\end{equation}
\end{theorem}
\begin{proof}
For the necessity, suppose that $H$ is a Hermitian operator providing $q_{\rm SEP}(\mathcal{E})$.
We also denote $\{M_{i}\}_{i=0}^{n}$ as a separable unambiguous measurement giving $p_{\rm SEP}(\mathcal{E})$.
From Conditions \eqref{eq:sepm} and \eqref{eq:hcon}, we have
\begin{equation}\label{eq:fnp1}
\mathrm{Tr}(M_{0}H)\geqslant0,~\mathrm{Tr}[M_{i}(H-\eta_{i}\rho_{i})]\geqslant0~\forall i=1,\ldots,n.
\end{equation}
We also note that
\begin{equation}\label{eq:fnp2}
\mathrm{Tr}(M_{0}H)+\sum_{i=1}^{n}\mathrm{Tr}[M_{i}(H-\eta_{i}\rho_{i})]
=\mathrm{Tr}H-\sum_{i=1}^{n}\eta_{i}\mathrm{Tr}(\rho_{i}M_{i})
=q_{\rm SEP}(\mathcal{E})-p_{\rm SEP}(\mathcal{E})=0,
\end{equation}
where the first equality follows from $\sum_{i=0}^{n}M_{i}=\mathbbm{1}$,
the second equality is due to the assumption of $H$ and $\{M_{i}\}_{i=0}^{n}$, and the last equality is by Theorem~\ref{thm:qupb}. 
Inequality~\eqref{eq:fnp1} and Eq.~\eqref{eq:fnp2} lead us to
Condition~\eqref{eq:nscqs}. Therefore, $\{M_{i}\}_{i=0}^{n}$ is a separable unambiguous measurement satisfying Condition \eqref{eq:nscqs}.\\
\indent For the sufficiency, we assume that $\{M_{i}\}_{i=0}^{n}$ is a separable unambiguous measurement and $H$ is a Hermitian operator satisfying Conditions~\eqref{eq:hcon} and \eqref{eq:nscqs}. This assumption implies
\begin{equation}\label{eq:fsp}
q_{\rm SEP}(\mathcal{E})
=p_{\rm SEP}(\mathcal{E})
\geqslant\sum_{i=1}^{n}\eta_{i}\mathrm{Tr}(\rho_{i}M_{i})
=\sum_{i=1}^{n}\eta_{i}\mathrm{Tr}(\rho_{i}M_{i})
+\mathrm{Tr}(M_{0}H)+\sum_{i=1}^{n}\mathrm{Tr}[M_{i}(H-\eta_{i}\rho_{i})]
=\mathrm{Tr}H
\geqslant q_{\rm SEP}(\mathcal{E}),
\end{equation}
where the first equality follows from Theorem~\ref{thm:qupb},
the second equality is from Condition~\eqref{eq:nscqs},
the last equality is due to $\sum_{i=0}^{n}M_{i}=\mathbbm{1}$, and 
the first and second inequalities are from the definitions of $p_{\rm SEP}(\mathcal{E})$ and $q_{\rm SEP}(\mathcal{E})$, respectively.
Inequality~\eqref{eq:fsp} leads us to $\mathrm{Tr}H=q_{\rm SEP}(\mathcal{E})$. Therefore, $H$ is a Hermitian operator giving $q_{\rm SEP}(\mathcal{E})$.
\end{proof}

\indent From Theorems \ref{thm:qupb} and \ref{thm:nsc}, we have the following corollary providing a necessary and sufficient condition on a separable unambiguous measurement $\{M_{i}\}_{i=0}^{n}$ to realize $p_{\rm SEP}(\mathcal{E})$.

\begin{corollary}\label{cor:qupb}
For a multipartite quantum state ensemble $\mathcal{E}=\{\eta_{i},\rho_{i}\}_{i=1}^{n}$, 
a separable unambiguous measurement $\{M_{i}\}_{i=0}^{n}$ gives $p_{\rm SEP}(\mathcal{E})$ if and only if there is a Hermitian operator $H$ satisfying Conditions~\eqref{eq:hcon} and \eqref{eq:nscqs}.
In this case, we have
\begin{equation}
p_{\rm SEP}(\mathcal{E})=\sum_{i=1}^{n}\eta_{i}\mathrm{Tr}(\rho_{i}M_{i})=\mathrm{Tr}H.
\end{equation}
\end{corollary}

Moreover, we have the following corollary that provides the relative ordering between $p_{\rm L}(\mathcal{E})$ and $q_{\rm SEP}(\mathcal{E})$.

\begin{corollary}\label{cor:plqg}
For a multipartite quantum state ensemble $\mathcal{E}=\{\eta_{i},\rho_{i}\}_{i=1}^{n}$, 
\begin{equation}
p_{\rm L}(\mathcal{E})\leqslant q_{\rm SEP}(\mathcal{E}),
\end{equation}
where the equality holds
if and only if there is a LOCC unambiguous measurement $\{M_{i}\}_{i=0}^{n}$ and a Hermitian operator $H$ satisfying Conditions~\eqref{eq:hcon} and \eqref{eq:nscqs}.
\end{corollary}
\begin{proof}
Since every LOCC measurement is separable, 
$p_{\rm L}(\mathcal{E})\leqslant p_{\rm SEP}(\mathcal{E})$.
Moreover, $p_{\rm L}(\mathcal{E})=p_{\rm SEP}(\mathcal{E})$ if and only if there is a LOCC unambiguous measurement realizing $p_{\rm SEP}(\mathcal{E})$. Thus, we can show from Theorem~\ref{thm:qupb} and Corollary~\ref{cor:qupb} that our statement is true.
\end{proof}

\indent Here, we provide the following example of two-qubit state ensemble $\mathcal{E}$ to illustrate 
how our results can be used to obtain
$p_{\rm G}(\mathcal{E})$, $q_{\rm SEP}(\mathcal{E})$, and $p_{\rm L}(\mathcal{E})$.\\

\noindent\textbf{Example 1.}
Let us consider 
the two-qubit state ensemble $\mathcal{E}=\{\eta_{i},\rho_{i}\}_{i=1}^{3}$ consisting of three product states with equal prior probabilities,
\begin{equation}\label{eq:dte}
\begin{array}{lll}
\eta_{1}=\frac{1}{3},~~\rho_{1}=|0\rangle\!\langle0|\otimes|0\rangle\!\langle0|,&\\[1mm]
\eta_{2}=\frac{1}{3},~~\rho_{2}=|\nu_{+}\rangle\!\langle\nu_{+}|\otimes|\nu_{+}\rangle\!\langle\nu_{+}|,~~
|\nu_{+}\rangle=\frac{1}{2}|0\rangle+\frac{\sqrt{3}}{2}|1\rangle,\\[1mm]
\eta_{3}=\frac{1}{3},~~\rho_{3}=|\nu_{-}\rangle\!\langle\nu_{-}|\otimes|\nu_{-}\rangle\!\langle\nu_{-}|,~~|\nu_{-}\rangle=\frac{1}{2}|0\rangle-\frac{\sqrt{3}}{2}|1\rangle.
\end{array}
\end{equation}
To obtain $p_{\rm G}(\mathcal{E})$, we use the unambiguous measurement $\{M_{i}\}_{i=0}^{3}$ with
\begin{equation}\label{eq:gopm}
\begin{array}{ll}
M_{0}=\frac{1}{2}|\Phi_{+}\rangle\!\langle\Phi_{+}|+|\Psi_{-}\rangle\!\langle\Psi_{-}|,\\[1mm]
M_{1}=\frac{5}{6}|\Phi_{1}\rangle\!\langle\Phi_{1}|,~|\Phi_{1}\rangle=\frac{3}{\sqrt{10}}|00\rangle-\frac{1}{\sqrt{10}}|11\rangle,\\[1mm]
M_{2}=\frac{5}{6}|\Phi_{2}\rangle\!\langle\Phi_{2}|,~|\Phi_{2}\rangle=\sqrt{\frac{3}{10}}|01\rangle+\sqrt{\frac{3}{10}}|10\rangle+\frac{2}{\sqrt{10}}|11\rangle,\\[1mm]
M_{3}=\frac{5}{6}|\Phi_{3}\rangle\!\langle\Phi_{3}|,~|\Phi_{3}\rangle=\sqrt{\frac{3}{10}}|01\rangle+\sqrt{\frac{3}{10}}|10\rangle-\frac{2}{\sqrt{10}}|11\rangle,
\end{array}
\end{equation}
and Hermitian operator
\begin{equation}\label{eq:gopk}
K=\frac{3}{8}|\Phi_{-}\rangle\!\langle\Phi_{-}|+\frac{3}{8}|\Psi_{+}\rangle\!\langle\Psi_{+}|,
\end{equation}
where
\begin{equation}
|\Phi_{\pm}\rangle=\frac{1}{\sqrt{2}}|00\rangle\pm\frac{1}{\sqrt{2}}|11\rangle,~
|\Psi_{\pm}\rangle=\frac{1}{\sqrt{2}}|01\rangle\pm\frac{1}{\sqrt{2}}|10\rangle
\end{equation}
are the Bell states in two-qubit systems.
We will show the unambiguous measurement $\{M_{i}\}_{i=0}^{3}$ and the Hermitian operator $K$ satisfy the conditions of Proposition~\ref{pro:nsc} for the ensemble $\mathcal{E}=\{\eta_{i},\rho_{i}\}_{i=1}^{3}$ in Eq.~\eqref{eq:dte}.\\
\indent For $K$ of Eq.~\eqref{eq:gopk},
Condition \eqref{eq:nscuda} is obvious and Condition \eqref{eq:nscudb} is also true due to the orthogonality of Bell states. For Condition \eqref{eq:nscudc}, we first note $|\Psi_{-}\rangle\!\langle\Psi_{-}|$ is orthogonal to each $\rho_{i}$, that is, $\mathrm{Tr}(|\Psi_{-}\rangle\!\langle\Psi_{-}|\rho_{i})=0$, $i=1,2,3$. Moreover, $|\Phi_{i}\rangle\!\langle\Phi_{i}|$ is orthogonal to $\rho_{j}$ for all $i,j=1,2,3$ with $i\neq j$. Thus, we have
\begin{equation}
\mathrm{Pos}_{i}(\mathcal{E})=\{E\succeq0\,|\,\mbox{ $E$ acting on the subspace spanned by $|\Phi_{i}\rangle$ and $|\Psi_{-}\rangle$}\}~~\forall i=1,2,3,
\end{equation}
for the ensemble $\mathcal{E}$ of Eq.~\eqref{eq:dte}.
Now, a straightforward calculation leads us to
\begin{equation}\label{eq:sclu}
\langle\Phi_{i}|(K-\eta_{i}\rho_{i})|\Phi_{i}\rangle
=\langle\Phi_{i}|(K-\eta_{i}\rho_{i})|\Psi_{-}\rangle
=\langle\Psi_{-}|(K-\eta_{i}\rho_{i})|\Psi_{-}\rangle=0~~\forall i=1,2,3,
\end{equation}
and this implies 
\begin{equation}\label{eq:trker}
\mathrm{Tr}[(K-\eta_{i}\rho_{i})E]=0
\end{equation}
for all $E$ in $\mathrm{Pos}_{i}(\mathcal{E})$.
From the definition of $\mathrm{Pos}_{i}^{*}(\mathcal{E})$ in Eq.~\eqref{eq:posis}, we have $K-\eta_{i}\rho_{i}\in\mathrm{Pos}_{i}^{*}(\mathcal{E})$ for each $i=1,2,3$, and this shows the validity of Condition \eqref{eq:nscudc}.
Finally, Condition \eqref{eq:nscudd} naturally follows from Eqs.~\eqref{eq:gopm} and \eqref{eq:sclu}.
From Proposition~\ref{pro:nsc}, the optimal success probability $p_{\rm G}(\mathcal{E})$ in Eq.~\eqref{eq:msp} is
\begin{equation}
p_{\rm G}(\mathcal{E})=\sum_{i=1}^{3}\eta_{i}\mathrm{Tr}(\rho_{i}M_{i})=\mathrm{Tr}K=\frac{3}{4}.
\end{equation}
\indent To obtain $q_{\rm SEP}(\mathcal{E})$ in Eq.~\eqref{eq:qsep}, let us consider
the Hermitian operator 
\begin{equation}\label{eq:exho1}
H=\frac{1}{2}|\Psi_{-}\rangle\!\langle\Psi_{-}|,
\end{equation}
and  the separable unambiguous measurement $\{M_{i}\}_{i=0}^{3}$ consisting of 
\begin{equation}\label{eq:elom}
\begin{array}{l}
M_{0}=\frac{4}{9}|1\rangle\!\langle1|\otimes|1\rangle\!\langle1|+\frac{4}{9}|\mu_{+}\rangle\!\langle\mu_{+}|\otimes|\mu_{+}\rangle\!\langle\mu_{+}|+\frac{4}{9}|\mu_{-}\rangle\!\langle\mu_{-}|\otimes|\mu_{-}\rangle\!\langle\mu_{-}|,\\[1mm]
M_{1}=\frac{4}{9}|\mu_{+}\rangle\!\langle\mu_{+}|\otimes|\mu_{-}\rangle\!\langle\mu_{-}|+\frac{4}{9}|\mu_{-}\rangle\!\langle\mu_{-}|\otimes|\mu_{+}\rangle\!\langle\mu_{+}|,\\[1mm]
M_{2}=\frac{4}{9}|\mu_{+}\rangle\!\langle\mu_{+}|\otimes|1\rangle\!\langle1|+\frac{4}{9}
|1\rangle\!\langle1|\otimes|\mu_{+}\rangle\!\langle\mu_{+}|,\\[1mm]
M_{3}=\frac{4}{9}|\mu_{-}\rangle\!\langle\mu_{-}|\otimes|1\rangle\!\langle1|+\frac{4}{9}
|1\rangle\!\langle1|\otimes|\mu_{-}\rangle\!\langle\mu_{-}|,
\end{array}
\end{equation}
where
\begin{equation}
|\mu_{\pm}\rangle=\frac{\sqrt{3}}{2}|0\rangle\pm\frac{1}{2}|1\rangle.
\end{equation}
We will show the Hermitian operator $H$ and the separable unambiguous measurement $\{M_{i}\}_{i=0}^{3}$ satisfy the conditions of Theorem~\ref{thm:nsc} for the ensemble $\mathcal{E}=\{\eta_{i},\rho_{i}\}_{i=1}^{3}$ in Eq.~\eqref{eq:dte}.\\
\indent For $H$ of Eq.~\eqref{eq:exho1},
Condition \eqref{eq:hcona} holds due to the argument after Eq.~\eqref{eq:sepd}
and Condition \eqref{eq:nscqsa} is also true from the fact that $|\Psi_{-}\rangle$ is orthogonal to $|1\rangle\otimes|1\rangle$,$|\mu_{+}\rangle\otimes|\mu_{+}\rangle$,$|\mu_{-}\rangle\otimes|\mu_{-}\rangle$. For Condition \eqref{eq:hconb}, we first note every positive-semidefinite product operator orthogonal to $\rho_{2}$ and $\rho_{3}$ is proportional to $|\mu_{+}\rangle\!\langle\mu_{+}|\otimes|\mu_{-}\rangle\!\langle\mu_{-}|$ or $|\mu_{-}\rangle\!\langle\mu_{-}|\otimes|\mu_{+}\rangle\!\langle\mu_{+}|$. Moreover, every positive-semidefinite product operator orthogonal to $\rho_{1}$ and $\rho_{3(2)}$ is proportional to $|\mu_{+(-)}\rangle\!\langle\mu_{+(-)}|\otimes|1\rangle\!\langle 1|$ or $|1\rangle\!\langle 1|\otimes|\mu_{+(-)}\rangle\!\langle\mu_{+(-)}|$.
Thus, we have
\begin{equation}
\begin{array}{l}
\mathrm{SEP}_{1}(\mathcal{E})=\{
a|\mu_{+}\rangle\!\langle\mu_{+}|\otimes|\mu_{-}\rangle\!\langle\mu_{-}|
+b|\mu_{-}\rangle\!\langle\mu_{-}|\otimes|\mu_{+}\rangle\!\langle\mu_{+}|~|~a,b\geqslant0\},\\[1mm]
\mathrm{SEP}_{2}(\mathcal{E})=\{
a|\mu_{+}\rangle\!\langle\mu_{+}|\otimes|1\rangle\!\langle 1|
+b|1\rangle\!\langle 1|\otimes|\mu_{+}\rangle\!\langle\mu_{+}|~|~a,b\geqslant0\},
\\[1mm]
\mathrm{SEP}_{3}(\mathcal{E})=\{
a|\mu_{-}\rangle\!\langle\mu_{-}|\otimes|1\rangle\!\langle 1|
+b|1\rangle\!\langle 1|\otimes|\mu_{-}\rangle\!\langle\mu_{-}|~|~a,b\geqslant0\},
\end{array}
\end{equation}
for the ensemble $\mathcal{E}$ of Eq.~\eqref{eq:dte}.
Now, a straightforward calculation leads us to
\begin{equation}\label{eq:sfclu}
\begin{array}{l}
\langle v|(H-\eta_{1}\rho_{1})|v\rangle=0
~~\forall|v\rangle=
|\mu_{+}\rangle\otimes|\mu_{-}\rangle,
|\mu_{-}\rangle\otimes|\mu_{+}\rangle,\\[1mm]
\langle v|(H-\eta_{2}\rho_{2})|v\rangle=0
~~\forall|v\rangle=
|\mu_{+}\rangle\otimes|1\rangle,
|1\rangle\otimes|\mu_{+}\rangle,
\\[1mm]
\langle v|(H-\eta_{3}\rho_{3})|v\rangle=0
~~\forall|v\rangle=
|\mu_{-}\rangle\otimes|1\rangle,
|1\rangle\otimes|\mu_{-}\rangle,
\end{array}
\end{equation}
and this implies 
\begin{equation}\label{eq:trher}
\mathrm{Tr}[(H-\eta_{i}\rho_{i})E]=0
\end{equation}
for all $E$ in $\mathrm{SEP}_{i}(\mathcal{E})$($i=1,2,3$).
From the definition of $\mathrm{SEP}_{i}^{*}(\mathcal{E})$($i=1,2,3$) in Eq.~\eqref{eq:sepd}, we have $H-\eta_{i}\rho_{i}\in\mathrm{SEP}_{i}^{*}(\mathcal{E})$ for each $i=1,2,3$, and this shows the validity of Condition \eqref{eq:hconb}.
Finally, Condition \eqref{eq:nscqsb} naturally follows from Eqs.~\eqref{eq:elom} and \eqref{eq:sfclu}.
From Theorem \ref{thm:nsc},
the minimum quantity $q_{\rm SEP}(\mathcal{E})$ in Eq.~\eqref{eq:qsep} is
\begin{equation}\label{eq:exqs}
q_{\rm SEP}(\mathcal{E})=\mathrm{Tr}H=\sum_{i=1}^{3}\eta_{i}\mathrm{Tr}(\rho_{i}M_{i})=\frac{1}{2}.
\end{equation}
\indent We also note that the measurement $\{M_{i}\}_{i=0}^{3}$ with Eq.~\eqref{eq:elom} is a LOCC measurement because it can be implemented by performing the same local measurement $\{\frac{2}{3}|1\rangle\langle1|,\frac{2}{3}|\mu_{+}\rangle\!\langle\mu_{+}|,\frac{2}{3}|\mu_{-}\rangle\!\langle\mu_{-}|\}$ on two subsystems. Thus, Corollary~\ref{cor:plqg} and Eq.~\eqref{eq:exqs} lead us to
\begin{equation}
p_{\rm L}(\mathcal{E})=q_{\rm SEP}(\mathcal{E})=\frac{1}{2}.
\end{equation}

\indent Now, we provide another example of mixed-state ensemble
in multipartite high-dimensional quantum systems to illustrate the application of 
our results in obtaining
$p_{\rm G}(\mathcal{E})$, $q_{\rm SEP}(\mathcal{E})$, and $p_{\rm L}(\mathcal{E})$.\\

\noindent\textbf{Example 2.}
For any integer $d\geqslant3$,
let us consider the following ($d-1$)-qu$d$it state
ensemble $\mathcal{E}=\{\eta_{i},\rho_{i}\}_{i=1}^{d}$ consisting of $d$ mixed states with equal prior probabilities,
\begin{gather}\label{eq:exens3}
\eta_{i}=\frac{1}{d},~\rho_{i}=
\frac{1}{d^{d-1}-2(d-1)}\Bigg[\mathbbm{1}-
\sum_{\substack{j=1\\ j\neq  i}}^{d}(|\Lambda_{j}\rangle\!\langle\Lambda_{j}|+|\Omega_{j}\rangle\!\langle\Omega_{j}|)\Bigg],
~i=1,\ldots,d,\\
|\Lambda_{j}\rangle=|j-1\rangle^{\otimes (d-1)},~
|\Omega_{j}\rangle=\frac{1}{\sqrt{d-1}}\sum_{\substack{k=0\\k\neq j-1}}^{d-1}
\bigotimes_{\substack{l=0\\k\oplus l\neq j-1}}^{d-1}|k\oplus l\rangle,
~~j=1,\ldots,d,\nonumber
\end{gather}
where $\oplus$ denotes modulo-$d$ addition.
To obtain $p_{\rm G}(\mathcal{E})$, we use the unambiguous measurement $\{M_{i}\}_{i=0}^{d}$ with
\begin{equation}\label{eq:gopm3}
M_{0}=\mathbbm{1}-\sum_{j=1}^{d}(|\Lambda_{j}\rangle\!\langle\Lambda_{j}|+
|\Omega_{j}\rangle\!\langle\Omega_{j}|),~
M_{i}=|\Lambda_{i}\rangle\!\langle\Lambda_{i}|+
|\Omega_{i}\rangle\!\langle\Omega_{i}|,~i=1,\ldots,d,
\end{equation}
and Hermitian operator
\begin{equation}\label{eq:gopk3}
K=\frac{1}{d^{d}-2d(d-1)}\sum_{j=1}^{d}(|\Lambda_{j}\rangle\!\langle\Lambda_{j}|+
|\Omega_{j}\rangle\!\langle\Omega_{j}|).
\end{equation}
We will show the unambiguous measurement $\{M_{i}\}_{i=0}^{d}$ and the Hermitian operator $K$ satisfy the conditions of Proposition~\ref{pro:nsc} for the ensemble $\mathcal{E}=\{\eta_{i},\rho_{i}\}_{i=1}^{d}$ in Eq.~\eqref{eq:exens3}.\\
\indent For $K$ of Eq.~\eqref{eq:gopk3},
Condition \eqref{eq:nscuda} is obvious and Condition \eqref{eq:nscudb} is also true due to the orthogonality of $\{|\Lambda_{i}\rangle,|\Omega_{i}\rangle\}_{i=1}^{d}$. 
For Condition \eqref{eq:nscudc}, we first note both $|\Lambda_{i}\rangle\!\langle\Lambda_{i}|$ and $|\Omega_{i}\rangle\!\langle\Omega_{i}|$ are orthogonal to $\rho_{j}$, that is, $\mathrm{Tr}(|\Lambda_{i}\rangle\!\langle\Lambda_{i}|\rho_{j})=\mathrm{Tr}(|\Omega_{i}\rangle\!\langle\Omega_{i}|\rho_{j})=0$, for all $i,j=1,\ldots,d$ with $i\neq j$.\\
\indent Let $O_{1}$ be the set of all positive semidefinite operators orthogonal to $\rho_{1}$. 
From the definition of $\rho_{i}$ in Eq.~\eqref{eq:exens3}, we have
\begin{equation}
O_{1}=\{E\succeq0\,|\,
\mbox{$E$ acting on the subspace spanned by $|\Lambda_{2}\rangle,\ldots,|\Lambda_{d}\rangle$ and $|\Omega_{2}\rangle,\ldots,|\Omega_{d}\rangle$}\}.
\end{equation}
Now, let $O_{j}$ be the set of all positive semidefinite operators orthogonal to $\rho_{j}$ for each $j=2,\ldots,d$. Similarly, we have
\begin{equation}
\begin{array}{l}
O_{j}=\{E\succeq0\,|\,
\mbox{$E$ acting on the subspace spanned by $|\Lambda_{1}\rangle,\ldots,\widehat{|\Lambda_{j}\rangle},\ldots,|\Lambda_{d}\rangle$ and $|\Omega_{1}\rangle,\ldots,\widehat{|\Omega_{j}\rangle},\ldots,|\Omega_{d}\rangle$}\},
\end{array}
\end{equation}
where
\begin{equation}\label{eq:lowh}
\begin{array}{l}
\{|\Lambda_{1}\rangle,\ldots,\widehat{|\Lambda_{j}\rangle},\ldots,|\Lambda_{d}\rangle\}
=\{
|\Lambda_{1}\rangle,\ldots,|\Lambda_{j-1}\rangle,|\Lambda_{j+1}\rangle,\ldots,|\Lambda_{d}\rangle\},\\[1mm]
\{|\Omega_{1}\rangle,\ldots,\widehat{|\Omega_{j}\rangle},\ldots,|\Omega_{d}\rangle\}
=
\{|\Omega_{1}\rangle,\ldots,|\Omega_{j-1}\rangle,|\Omega_{j+1}\rangle,\ldots,|\Omega_{d}\rangle\}.
\end{array}
\end{equation}
From the definition of $\mathrm{Pos}_{i}(\mathcal{E})$ in Eq.~\eqref{eq:posi}, we have
\begin{equation}\label{eq:expos}
\mathrm{Pos}_{i}(\mathcal{E})
=\bigcap_{\substack{j=1\\ j\neq i}}^{d}O_{j}
=\{E\succeq0\,|\,\mbox{ $E$ acting on the subspace spanned by $|\Lambda_{i}\rangle$ and $|\Omega_{i}\rangle$}\}~~\forall i=1,\ldots,d,
\end{equation}
for the ensemble $\mathcal{E}$ of Eq.~\eqref{eq:exens3}.
Moreover, a straightforward calculation leads us to
\begin{equation}\label{eq:stcale}
\langle\Lambda_{i}|(K-\eta_{i}\rho_{i})|\Lambda_{i}\rangle
=\langle\Lambda_{i}|(K-\eta_{i}\rho_{i})|\Omega_{i}\rangle
=\langle\Omega_{i}|(K-\eta_{i}\rho_{i})|\Omega_{i}\rangle=0
~\forall i=1,\ldots,d.
\end{equation}
This implies Eq.~\eqref{eq:trker}
for all $E$ in $\mathrm{Pos}_{i}(\mathcal{E})$($i=1,\ldots,d$).
From the definition of $\mathrm{Pos}_{i}^{*}(\mathcal{E})$ in Eq.~\eqref{eq:posis}, we have $K-\eta_{i}\rho_{i}\in\mathrm{Pos}_{i}^{*}(\mathcal{E})$ for each $i=1,\ldots,d$, and this shows the validity of Condition \eqref{eq:nscudc}.
Finally, Condition \eqref{eq:nscudd} naturally follows from Eqs.~\eqref{eq:gopm3} and \eqref{eq:stcale}.
Therefore, the optimal success probability $p_{\rm G}(\mathcal{E})$ in Eq.~\eqref{eq:msp} is
\begin{equation}\label{eq:exosp2}
p_{\rm G}(\mathcal{E})=\sum_{i=1}^{d}\eta_{i}\mathrm{Tr}(\rho_{i}M_{i})=\mathrm{Tr}K=\frac{2}{d^{d-1}-2(d-1)}.
\end{equation}
\indent To obtain $q_{\rm SEP}(\mathcal{E})$ in Eq.~\eqref{eq:qsep}, let us consider 
the Hermitian operator
\begin{equation}\label{eq:elok3}
H=\frac{1}{d^{d}-2d(d-1)}\sum_{j=1}^{d}|\Lambda_{j}\rangle\!\langle\Lambda_{j}|,
\end{equation}
and the separable unambiguous measurement $\{M_{i}\}_{i=0}^{d}$ consisting of 
\begin{equation}\label{eq:elom3}
M_{0}=\mathbbm{1}-\sum_{j=1}^{d}|\Lambda_{j}\rangle\!\langle\Lambda_{j}|,~
M_{i}=|\Lambda_{i}\rangle\!\langle\Lambda_{i}|,~i=1,\ldots,d,
\end{equation}
where $\{|\Lambda_{j}\rangle\}_{j=1}^{d}$ is defined in Eq.~\eqref{eq:exens3}.
We will show the Hermitian operator $H$ and the separable unambiguous measurement $\{M_{i}\}_{i=0}^{d}$ satisfy the conditions of Theorem~\ref{thm:nsc} for the ensemble $\mathcal{E}=\{\eta_{i},\rho_{i}\}_{i=1}^{d}$ in Eq.~\eqref{eq:exens3}.\\
\indent For $H$ of Eq.~\eqref{eq:elok3},
Condition \eqref{eq:hcona} is satisfied from the argument after Eq.~\eqref{eq:sepd}
and Condition \eqref{eq:nscqsa} is also true due to the orthogonality of $\{|\Lambda_{i}\rangle\}_{i=1}^{d}$. For Condition \eqref{eq:hconb}, we first note that $|\Lambda_{i}\rangle\!\langle\Lambda_{i}|$ is obviously separable for all $i=1,\ldots,d$.
However, $|\Omega_{i}\rangle\!\langle\Omega_{i}|$ is not separable for each $i=1,\ldots,d$ because the reduced density operator of $|\Omega_{i}\rangle\!\langle\Omega_{i}|$ onto any qu$d$it subsystem is not rank one.\\
\indent 
For each $i=1,\ldots,d$, we also note that $|\Lambda_{i}\rangle\!\langle\Lambda_{i}|$
is the only pure product state acting on the subspace spanned by $|\Lambda_{i}\rangle$ and $|\Omega_{i}\rangle$.
To see this, let us consider any pure state $(c_{1}|\Lambda_{i}\rangle+c_{2}|\Omega_{i}\rangle)(c_{1}^{*}\langle\Lambda_{i}|+c_{2}^{*}\langle\Omega_{i}|)$ on the subspace spanned by $|\Lambda_{i}\rangle$ and $|\Omega_{i}\rangle$ with complex numbers $c_{1}$ and $c_{2}$ such that $|c_{1}|^{2}+|c_{2}|^{2}=1$. 
Due to the definitions of $|\Lambda_{i}\rangle$ and $|\Omega_{i}\rangle$ in Eq.~\eqref{eq:exens3}, it is straightforward to show that any partial trace of $|\Lambda_{i}\rangle\!\langle\Omega_{i}|$ is the zero operator.
Thus, the reduced density operator of $(c_{1}|\Lambda_{i}\rangle+c_{2}|\Omega_{i}\rangle)(c_{1}^{*}\langle\Lambda_{i}|+c_{2}^{*}\langle\Omega_{i}|)$ is
\begin{equation}\label{eq:rdo}
\mathrm{Tr}_{\mathbb{S}}[(c_{1}|\Lambda_{i}\rangle+c_{2}|\Omega_{i}\rangle)(c_{1}^{*}\langle\Lambda_{i}|+c_{2}^{*}\langle\Omega_{i}|)]
=|c_{1}|^{2}\,|\mathrm{Tr}_{\mathbb{S}}(|\Lambda_{i}\rangle\!\langle\Lambda_{i}|)+|c_{2}|^{2}\,|\mathrm{Tr}_{\mathbb{S}}(|\Omega_{i}\rangle\!\langle\Omega_{i}|),
\end{equation}
where $\mathbb{S}$ is any $\alpha$-qu$d$it subsystem 
with $\alpha<d-1$. If $c_{2}$ is nonzero, then the reduced density operator in Eq.~\eqref{eq:rdo} is not rank one because $|\Omega_{i}\rangle\!\langle\Omega_{i}|$ is not separable, which implies that $(c_{1}|\Lambda_{i}\rangle+c_{2}|\Omega_{i}\rangle)(c_{1}^{*}\langle\Lambda_{i}|+c_{2}^{*}\langle\Omega_{i}|)$ is not separable.
Thus, $|\Lambda_{i}\rangle\!\langle\Lambda_{i}|$ is the only pure product state acting on the subspace spanned by $|\Lambda_{i}\rangle$ and $|\Omega_{i}\rangle$.\\
\indent 
Now, let us consider $\mathrm{SEP}_{i}(\mathcal{E})$
of the ensemble $\mathcal{E}$ in Eq.~\eqref{eq:exens3} for each $i=1,\ldots,d$.
From Eq.~\eqref{eq:expos} and the definition of $\mathrm{SEP}_{i}(\mathcal{E})$ in Eq.~\eqref{eq:sepi}, we have
\begin{equation}
\mathrm{SEP}_{i}(\mathcal{E})
=\{E\in\mathrm{SEP}\,|\,\mbox{ $E$ acting on the subspace spanned by $|\Lambda_{i}\rangle$ and $|\Omega_{i}\rangle$}\}.
\end{equation}
Because any separable state is a convex combination of pure product states and $|\Lambda_{i}\rangle\!\langle\Lambda_{i}|$ is the only pure product state acting on the subspace spanned by $|\Lambda_{i}\rangle$ and $|\Omega_{i}\rangle$, we have
\begin{equation}
\mathrm{SEP}_{i}(\mathcal{E})
=\big\{a|\Lambda_{i}\rangle\!\langle\Lambda_{i}|~\big|~a\geqslant0\big\}.
\end{equation}
Moreover, a straightforward calculation leads us to
\begin{equation}\label{eq:scluth}
\langle\Lambda_{i}|(H-\eta_{i}\rho_{i})|\Lambda_{i}\rangle=0,
\end{equation}
which implies $\mathrm{Tr}[(H-\eta_{i}\rho_{i})E]=0$
for all $E$ in $\mathrm{SEP}_{i}(\mathcal{E})$.
Thus, we have $H-\eta_{i}\rho_{i}\in\mathrm{SEP}_{i}^{*}(\mathcal{E})$, and this shows the validity of Condition \eqref{eq:hconb}
for each $i=1,\ldots,d$.
Finally, Condition \eqref{eq:nscqsb} naturally follows from Eqs.~\eqref{eq:elom3} and \eqref{eq:scluth}.\\
\indent From Theorem~\ref{thm:nsc},
the minimum quantity $q_{\rm SEP}(\mathcal{E})$ in Eq.~\eqref{eq:qsep} is
\begin{equation}\label{eq:exqs3}
q_{\rm SEP}(\mathcal{E})=\mathrm{Tr}H=\sum_{i=1}^{d}\eta_{i}\mathrm{Tr}(\rho_{i}M_{i})=\frac{1}{d^{d-1}-2(d-1)}.
\end{equation}
Moreover, the measurement $\{M_{i}\}_{i=0}^{d}$ with Eq.~\eqref{eq:elom3} is a LOCC measurement since it can be implemented by performing the same local measurement $\{|i\rangle\!\langle i|\}_{i=0}^{d-1}$ on all subsystems. Thus, Corollary~\ref{cor:plqg} and Eq.~\eqref{eq:exqs3} lead us to
\begin{equation}\label{eq:plqse}
p_{\rm L}(\mathcal{E})=q_{\rm SEP}(\mathcal{E})=\frac{1}{d^{d-1}-2(d-1)}.
\end{equation}

\indent Examples~1 and 2 are special cases that $p_{\rm L}(\mathcal{E})=p_{\rm SEP}(\mathcal{E})$, or equivalently $p_{\rm L}(\mathcal{E})=q_{\rm SEP}(\mathcal{E})$.
We also note that there exist separable-state ensembles with $p_{\rm L}(\mathcal{E})<q_{\rm SEP}(\mathcal{E})$, therefore $p_{\rm L}(\mathcal{E})<p_{\rm SEP}(\mathcal{E})$.
A well-known example with $p_{\rm L}(\mathcal{E})<p_{\rm SEP}(\mathcal{E})$ is the domino state ensemble which can be perfectly discriminated by separable measurements but not by LOCC measurements \cite{benn19991}.

\section*{DISCUSSION}\label{dissec}
\indent In this paper, we have considered the situation of unambiguously discriminating multipartite quantum states, and provided an upper bound $q_{\rm SEP}(\mathcal{E})$ for the maximum success probability of optimal local discrimination $p_{\rm L}(\mathcal{E})$.
We have further established a necessary and sufficient condition
for the Hermitian operator $H$ to realize $q_{\rm SEP}(\mathcal{E})$.
Moreover, we have provided a necessary and sufficient condition for the upper bound $q_{\rm SEP}(\mathcal{E})$ to be saturated.
Finally, we have illustrated our results by examples in multidimensional multipartite quantum systems.\\
\indent We remark that finding $p_{\rm G}(\mathcal{E})$ and $q_{\rm SEP}(\mathcal{E})$ in unambiguously discriminating separable quantum states can be useful in studying the phenomenon of \emph{nonlocality without entanglement}(NLWE)\cite{benn19991}.
For the optimal UD of a separable-state ensemble $\{\eta_{i},\rho_{i}\}_{i=1}^{n}$,
the NLWE phenomenon occurs when $p_{\rm G}(\mathcal{E})$ cannot be realized only by LOCC,
that is, $p_{\rm L}(\mathcal{E})<p_{\rm G}(\mathcal{E})$. Due to Corollary~\ref{cor:plqg}, $q_{\rm SEP}(\mathcal{E})<p_{\rm G}(\mathcal{E})$ means $p_{\rm L}(\mathcal{E})<p_{\rm G}(\mathcal{E})$, therefore the occurrence of NLWE. It is a natural future work to find good bounds on optimal local discrimination in other generalized state discrimination strategies such as an optimal discrimination with a fixed rate of inconclusive results\cite{chef1998,zhan1999,fiur2003,baga2012,herz2015}.

\section*{METHODS}\label{mtdsec}
\indent In this section, we prove Theorem~\ref{thm:qupb} by showing that 
\begin{subequations}
\begin{gather}
p_{\rm SEP}(\mathcal{E})\leqslant q_{\rm SEP}(\mathcal{E}),\label{eq:pleqq}\\
p_{\rm SEP}(\mathcal{E})\geqslant q_{\rm SEP}(\mathcal{E}).\label{eq:pgeqq}
\end{gather}
\end{subequations}

\subsection*{Proof of Inequality~\eqref{eq:pleqq}}
\indent Let us assume that $\{M_{i}\}_{i=0}^{n}$ is a separable unambiguous measurement realizing $p_{\rm SEP}(\mathcal{E})$ and $H$ is a Hermitian operator giving $q_{\rm SEP}(\mathcal{E})$.
Since this assumption implies Conditions \eqref{eq:sepm} and \eqref{eq:hcon}, we have
\begin{equation}\label{eq:twcc}
\mathrm{Tr}(M_{0}H)\geqslant0,~\mathrm{Tr}[M_{i}(H-\eta_{i}\rho_{i})]\geqslant0~\forall i=1,\ldots,n,
\end{equation}
which lead us to 
\begin{equation}\label{eq:pleqh}
p_{\rm SEP}(\mathcal{E})=\sum_{i=1}^{n}\eta_{i}\mathrm{Tr}(\rho_{i}M_{i})
\leqslant\sum_{i=1}^{n}\eta_{i}\mathrm{Tr}(\rho_{i}M_{i})
+\mathrm{Tr}(M_{0}H)+\sum_{i=1}^{n}\mathrm{Tr}[M_{i}(H-\eta_{i}\rho_{i})]=\mathrm{Tr}H=q_{\rm SEP}(\mathcal{E}),
\end{equation}
where the second equality is due to $\sum_{i=0}^{n}M_{i}=\mathbbm{1}$. 
Therefore, Inequality~\eqref{eq:pleqq} holds.

\subsection*{Proof of Inequality~\eqref{eq:pgeqq}}
We first prove Inequality~\eqref{eq:pgeqq} when $p_{\rm SEP}(\mathcal{E})=0$. In this case, we claim that every $M\in\mathrm{SEP}_{j}(\mathcal{E})$ satisfies $\mathrm{Tr}(\rho_{j}M)=0$ for each $j=1,\ldots,n$.
To see this, suppose that there is a positive-semidefinite product operator $E\in\mathrm{SEP}_{j}(\mathcal{E})$ with $\mathrm{Tr}(\rho_{j}E)>0$. Since $\mathbbm{1}-E/\mathrm{Tr}E$ is obvious separable,  the following measurement $\{M_{i}\}_{i=0}^{n}$ 
is a separable unambiguous measurement:
\begin{equation}\label{eq:fpsum}
M_{0}=\mathbbm{1}-\frac{1}{\mathrm{Tr}E}E,~~M_{j}=\frac{1}{\mathrm{Tr}E}E,~~ M_{i}=0_{\mathcal{H}}~~\forall i=1,\ldots,n~~\mbox{with}~~i\neq j,
\end{equation}
where $0_{\mathcal{H}}$ is the zero operator on $\mathcal{H}$.\\
\indent Moreover, the measurement of \eqref{eq:fpsum} gives 
\begin{equation}\label{eq:simm}
\sum_{i=1}^{n}\eta_{i}\mathrm{Tr}(\rho_{i}M_{i})=
\eta_{j}\mathrm{Tr}(\rho_{j}E)>0.
\end{equation}
From Inequality~\eqref{eq:simm} and the definition of $p_{\rm SEP}(\mathcal{E})$, we have $p_{\rm SEP}(\mathcal{E})>0$, a contradiction. Therefore, 
there is no positive-semidefinite product operator $E\in\mathrm{SEP}_{j}(\mathcal{E})$ with $\mathrm{Tr}(\rho_{j}E)>0$.
Since every positive-semidefinite separable operator can be represented as a summation of positive-semidefinite product operators, we have
\begin{equation}\label{eq:trrjm}
\mathrm{Tr}(\rho_{j}M)=0
\end{equation}
for all $M\in\mathrm{SEP}_{j}(\mathcal{E})$.
Equation~\eqref{eq:trrjm} together with the definition of $\mathrm{SEP}_{j}^{*}(\mathcal{E})$ in Eq.~\eqref{eq:sepd} imply that 
\begin{equation}\label{eq:arim}
a\rho_{j}\in\mathrm{SEP}_{j}^{*}(\mathcal{E})
\end{equation}
for any real number $a$ and $j \in \left\{1,\ldots,n\right\}$.\\
\indent By letting $H =0_{\mathcal{H}}$, we trivially have
\begin{equation}\label{eq:himss}
H\in\mathrm{SEP}^{*}.
\end{equation}
For each $i=1,\ldots,n$, we also have
\begin{equation}\label{eq:hmree}
H-\eta_{i}\rho_{i}=-\eta_{i}\rho_{i}\in\mathrm{SEP}_{i}^{*}(\mathcal{E})
\end{equation}
where the inclusion is from Eq.~\eqref{eq:arim}.
Equations~\eqref{eq:himss} and \eqref{eq:hmree} imply that  $H=0_{\mathcal{H}}$ satisfies Condition~\eqref{eq:hcon}, therefore
\begin{equation}\label{eq:qshzp}
q_{\rm SEP}(\mathcal{E})
\leqslant\mathrm{Tr}H=0,
\end{equation}
where the inequality is due to the definition of $q_{\rm SEP}(\mathcal{E})$. Thus, Inequality~\eqref{eq:qshzp} and the assumption $p_{\rm SEP}(\mathcal{E})=0$ lead us to Inequality~\eqref{eq:pgeqq}.\\
\indent Now, we prove Inequality~\eqref{eq:pgeqq} when $p_{\rm SEP}(\mathcal{E})>0$.
\begin{lemma}\label{lem:trn}
If $E\in\mathrm{SEP}^{*}$ and $E\neq0_{\mathcal{H}}$, then
$\mathrm{Tr}E>0$, where $0_{\mathcal{H}}$ is the zero operator on $\mathcal{H}$.
\end{lemma}
\begin{proof}
The proof is by contradiction.
We first note that $\mathrm{Tr}E=\mathrm{Tr}(\mathbbm{1}E)\geqslant0$ because $E\in\mathrm{SEP}^{*}$ and the identity operator $\mathbbm{1}$ is obviously separable.
Thus, let us suppose $\mathrm{Tr}E=0$.\\
\indent For an arbitrary orthonormal product basis $\{|e_{i}\rangle\}_{i=1}^{D}$ of the multipartite Hilbert space $\mathcal{H}=\bigotimes_{k=1}^{m}\mathcal{H}_{k}$, we have
\begin{equation}\label{eq:tree2}
\sum_{i=1}^{D}\mathrm{Tr}(E|e_{i}\rangle\!\langle e_{i}|)=\mathrm{Tr}(E\mathbbm{1})
=\mathrm{Tr}E=0,
\end{equation}
where $D$ is the dimension of $\mathcal{H}$. 
From $E\in\mathrm{SEP}^{*}$ and $|e_{i}\rangle\!\langle e_{i}|\in\mathrm{SEP}$ for all $i=1,\ldots,D$, we have
\begin{equation}\label{eq:tree1}
\mathrm{Tr}(E|e_{i}\rangle\!\langle e_{i}|)\geqslant0~~\forall i=1,\ldots,D.
\end{equation}
Equation~\eqref{eq:tree2} and Inequality~\eqref{eq:tree1} lead us to
$\mathrm{Tr}(E|e_{i}\rangle\!\langle e_{i}|)=0$
for all $i=1,\ldots,D$.
Since the choice of $\{|e_{i}\rangle\}_{i=1}^{D}$ can be arbitrary, $\mathrm{Tr}(E|e\rangle\!\langle e|)=0$ for any product vector $|e\rangle\in\mathcal{H}$, therefore
\begin{equation}\label{eq:treff}
\mathrm{Tr}(EF)=0~~\forall F\in\mathrm{SEP}.
\end{equation}
\indent We note that $\mathrm{SEP}$ spans the set of all Hermitian operators on $\mathcal{H}$. To see this, we first note that the set of all positive-semidefinite operators on $\mathcal{H}_{k}$ spans the set of all Hermitian operators on $\mathcal{H}_{k}$ for each $k=1,\ldots,m$.
Moreover, every Hermitian operator $A$ on $\mathcal{H}$ can be represented as a summation of product Hermitian operators, 
\begin{equation}
A=\sum_{l}\bigotimes_{k=1}^{m}A_{l,k},
\end{equation}
where $A_{l,k}$ is a Hermitian operator on $\mathcal{H}_{k}$ for each $k=1,\ldots,m$.
Therefore, Eq.~\eqref{eq:treff} leads us to $\mathrm{Tr}(EF)=0$ for any Hermitian operator $F$ on $\mathcal{H}$. This means $E=0_{\mathcal{H}}$, a contradiction. Thus, $\mathrm{Tr}E>0$.
\end{proof}
\indent Let us consider the set
\begin{equation}\label{eq:sesm}
S(\mathcal{E})=\big\{\big(\sum_{i=1}^{n}\eta_{i}\mathrm{Tr}(\rho_{i}M_{i})-p,\mathbbm{1}-\sum_{i=0}^{n}M_{i}\big)\in\mathbb{R}\times\mathbb{H}\,\big|\,
p>p_{\rm SEP}(\mathcal{E}),~
M_{0}\in\mathrm{SEP},~M_{i}\in\mathrm{SEP}_{i}(\mathcal{E})~\forall i=1,\ldots,n\big\},
\end{equation}
where $\mathbb{R}$ is the set of all real numbers and $\mathbb{H}$ is the set of all Hermitian operators on the multipartite Hilbert space $\mathcal{H}$.
We note that $S(\mathcal{E})$ is a convex set due to the convexity of $\mathrm{SEP}$ and $\mathrm{SEP}_{i}(\mathcal{E})$($i=1,\ldots,n$) in Eqs.~\eqref{eq:sep} and \eqref{eq:sepi}.
Moreover, $S(\mathcal{E})$ does not have the origin $(0,0_{\mathcal{H}})$ of $\mathbb{R}\times\mathbb{H}$ otherwise there exists a separable unambiguous measurement $\{M_{i}\}_{i=0}^{n}$ with $\sum_{i=1}^{n}\eta_{i}\mathrm{Tr}(\rho_{i}M_{i})>p_{\rm SEP}(\mathcal{E})$, and this contradicts the optimality of $p_{\rm SEP}(\mathcal{E})$ in Eq.~\eqref{eq:psep}. We also note that the Cartesian product $\mathbb{R}\times\mathbb{H}$
can be considered as a real vector space with an inner product defined as
\begin{equation}\label{eq:inpr}
\langle (a,A),(b,B)\rangle=ab+\mathrm{Tr}(AB),
~~(a,A),(b,B)\in\mathbb{R}\times\mathbb{H}.
\end{equation}
\indent Since $S(\mathcal{E})$ in Eq.~\eqref{eq:sesm} and the  single-element set $\{(0,0_{\mathcal{H}})\}$ are disjoint convex sets, it follows from separating hyperplane theorem\cite{boyd2004,sht} that there is $(\gamma,\Gamma)\in\mathbb{R}\times\mathbb{H}$ such that 
\begin{subequations}
\begin{gather}
(\gamma,\Gamma)\neq(0,0_{\mathcal{H}}),\label{eq:rgrg1}\\
\langle (\gamma,\Gamma),(r,G)\rangle\leqslant0~~\forall(r,G)\in S(\mathcal{E})\label{eq:rgrg2}.
\end{gather}
\end{subequations}
\indent Suppose 
\begin{subequations}
\begin{gather}
\mathrm{Tr}\Gamma\leqslant\gamma p_{\rm SEP}(\mathcal{E})\label{eq:trglg},\\
\Gamma\in\mathrm{SEP}^{*},\label{eq:gimss}\\
\Gamma-\gamma\eta_{i}\rho_{i}\in\mathrm{SEP}_{i}^{*}(\mathcal{E})~\forall i=1,\ldots,n,\label{eq:gsgmg}\\
\gamma>0,\label{eq:gleq0}
\end{gather}
\end{subequations}
where $\mathrm{SEP}^{*}$ and $\mathrm{SEP}_{i}^{*}(\mathcal{E})$ are defined in Eq.~\eqref{eq:sepd}.
From Conditions~\eqref{eq:gimss}, \eqref{eq:gsgmg}, and \eqref{eq:gleq0}, the Hermitian operator $H=\Gamma/\gamma$ satisfies Condition~\eqref{eq:hcon}.
Thus, the definition of $q_{\rm SEP}(\mathcal{E})$ in Eq.~\eqref{eq:qsep} leads us to
\begin{equation}\label{eq:qslth}
q_{\rm SEP}(\mathcal{E})\leqslant\mathrm{Tr}H.
\end{equation}
Moreover, Condition~\eqref{eq:trglg} and the definition of $H=\Gamma/\gamma$ imply
\begin{equation}\label{eq:thlps}
\mathrm{Tr}H\leqslant p_{\rm SEP}(\mathcal{E}).
\end{equation}
Inequalities~\eqref{eq:qslth} and \eqref{eq:thlps} complete the proof of Inequality~\eqref{eq:pgeqq}.\\
\indent The rest of this section is to prove Conditions~\eqref{eq:trglg}, \eqref{eq:gimss}, \eqref{eq:gsgmg}, and \eqref{eq:gleq0}.
\begin{proof}[Proof of \eqref{eq:trglg}]
From Eq.~\eqref{eq:inpr}, Inequality~\eqref{eq:rgrg2} can be rewritten as
\begin{equation}\label{eq:trlerp}
\mathrm{Tr}\Gamma-\mathrm{Tr}(M_{0}\Gamma)-\sum_{i=1}^{n}\mathrm{Tr}[M_{i}(\Gamma-\gamma\eta_{i}\rho_{i})]\leqslant \gamma p
\end{equation}
for all $p>p_{\rm SEP}(\mathcal{E})$ and all $\{M_{i}\}_{i=0}^{n}$ satisfying Condition~\eqref{eq:sepm}. If $M_{i}=0_{\mathcal{H}}$ for all $i=0,1,\ldots,n$,  Inequality~\eqref{eq:trlerp} becomes Inequality~\eqref{eq:trglg}
by taking the limit of $p$ to $p_{\rm SEP}(\mathcal{E})$.
\end{proof}
\begin{proof}[Proof of \eqref{eq:gimss}]
For an arbitrary $M_{0}\in\mathrm{SEP}$ and $M_{i}=0$ for all $i=1,\ldots,n$, $\{M_{i}\}_{i=0}^{n}$ clearly satisfies Condition~\eqref{eq:sepm}. 
In this case, Inequality~\eqref{eq:trlerp} becomes
\begin{equation}\label{eq:jgmgg}
\mathrm{Tr}\Gamma-\mathrm{Tr}(M_{0}\Gamma)\leqslant\gamma p_{\rm SEP}(\mathcal{E})
\end{equation}
by taking the limit of $p$ to 
$p_{\rm SEP}(\mathcal{E})$.\\
\indent Suppose $\Gamma\notin\mathrm{SEP}^{*}$,
then there exists $M\in\mathrm{SEP}$ with $\mathrm{Tr}(M\Gamma)<0$. We note that $M\in\mathrm{SEP}$ implies $tM\in\mathrm{SEP}$ for any $t>0$. Thus, $\{M_{i}\}_{i=0}^{n}$ with $M_{0}=tM$ for $t>0$ and $M_{i}=0$ for all $i=1,\ldots,n$ also satisfies Condition~\eqref{eq:sepm}.
Now, Inequality~\eqref{eq:jgmgg} can be rewritten as
\begin{equation}\label{eq:tmttm}
\mathrm{Tr}\Gamma-\mathrm{Tr}(tM\Gamma)\leqslant\gamma p_{\rm SEP}(\mathcal{E}).
\end{equation}
Since Inequality~\eqref{eq:tmttm} is true for arbitrary large $t>0$, $\gamma p_{\rm SEP}(\mathcal{E})$ can also be arbitrary large. However, this contradicts that both  $\gamma$ and $p_{\rm SEP}(\mathcal{E})$ are finite. Thus, $\Gamma\in\mathrm{SEP}^{*}$, which completes the proof of \eqref{eq:gimss}.
\end{proof}

\begin{proof}[Proof of \eqref{eq:gsgmg}]
The proof method is analogous to that of \eqref{eq:gimss}.
For each $j\in\{1,\ldots,n\}$,
let us consider an arbitrary $M_{j}\in\mathrm{SEP}_{j}(\mathcal{E})$ and $M_{i}=0_{\mathcal{H}}$ for all $i=0,1,\ldots,n$ with $i\neq j$. In this case,
$\{M_{i}\}_{i=0}^{n}$ clearly satisfies Condition~\eqref{eq:sepm} and Inequality~\eqref{eq:trlerp} becomes
\begin{equation}\label{eq:jgmgg2}
\begin{array}{ll}
\mathrm{Tr}\Gamma-\mathrm{Tr}[M_{j}(\Gamma-\gamma\eta_{j}\rho_{j})]\leqslant\gamma p_{\rm SEP}(\mathcal{E})
\end{array}
\end{equation}
by taking the limit of $p$ to 
$p_{\rm SEP}(\mathcal{E})$.\\
\indent Suppose $\Gamma-\gamma\eta_{j}\rho_{j}\notin\mathrm{SEP}_{j}^{*}(\mathcal{E})$,
then there exists $M\in\mathrm{SEP}_{j}(\mathcal{E})$ with $\mathrm{Tr}[M(\Gamma-\gamma\eta_{j}\rho_{j})]<0$. We note that $M\in\mathrm{SEP}_{j}(\mathcal{E})$ implies $tM\in\mathrm{SEP}_{j}(\mathcal{E})$ for any $t>0$. Thus, $\{M_{i}\}_{i=0}^{n}$ consisting of $M_{j}=tM$ for $t>0$ and $M_{i}=0$ for all $i=0,1,\ldots,n$ with $i\neq j$ also satisfies Condition~\eqref{eq:sepm}.
Now, Inequality~\eqref{eq:jgmgg2} can be rewritten as
\begin{equation}\label{eq:jgmsep}
\begin{array}{ll}
\mathrm{Tr}\Gamma-\mathrm{Tr}[tM(\Gamma-\gamma\eta_{j}\rho_{j})]\leqslant\gamma p_{\rm SEP}(\mathcal{E})
\end{array}
\end{equation}
Since Inequality~\eqref{eq:jgmsep} is true for arbitrary large $t>0$, $\gamma p_{\rm SEP}(\mathcal{E})$ can also be arbitrary large. However, this contradicts that both  $\gamma$ and $p_{\rm SEP}(\mathcal{E})$ are finite. Thus, $\Gamma-\gamma\eta_{j}\rho_{j}\in\mathrm{SEP}_{j}^{*}(\mathcal{E})$, which completes the proof of \eqref{eq:gsgmg}.
\end{proof}

\begin{proof}[Proof of \eqref{eq:gleq0}]
\indent Suppose $\Gamma\neq0_{\mathcal{H}}$.
From Lemma~\ref{lem:trn} together with Inequality~\eqref{eq:gimss}, we have
$\mathrm{Tr}\Gamma>0$. Thus,
Inequality~\eqref{eq:trglg} and the fact that $\mathrm{Tr}\Gamma>0$ guarantee $\gamma>0$.\\ 
\indent Now, suppose $\Gamma=0_{\mathcal{H}}$. We have $\gamma\neq0$, otherwise a contradiction to Condition \eqref{eq:rgrg1}. The strict positivity of $\gamma$ follows from Inequality~\eqref{eq:trglg} with $\mathrm{Tr}\Gamma=0$ and $p_{\rm SEP}(\mathcal{E})>0$. Thus, Inequality~\eqref{eq:gleq0} holds regardless of $\Gamma$.
\end{proof}

\section*{ACKNOWLEDGEMENTS}
\noindent
This work was supported by Basic Science Research Program(NRF-2020R1F1A1A010501270) and Quantum Computing Technology Development Program(NRF-2020M3E4A1080088) through the National Research Foundation of Korea(NRF) grant funded by the Korea government(Ministry of Science and ICT).



\begin{thebibliography}{61}%
\bibitem{chit20142}  
Chitambar,~E., Leung,~D., Mančinska,~L., Ozols,~M. \& Winter,~A. 
Everything you always wanted to know about LOCC (but were afraid to ask). 
\textit{Commun. Math. Phys.} 
\textbf{328}, 303 (2014).
%
\bibitem{horo2009}  
Horodecki, R., Horodecki, P., Horodecki, M. \& Horodecki, K.
Quantum entanglement. 
\textit{Rev. Mod. Phys.} 
\textbf{81}, 865 (2009).
%
\bibitem{chit2019}
Chitambar,~E. \& Gour, G.,
Quantum resource theories, 
\textit{Rev. Mod. Phys.} 
\textbf{91}, 025001 (2019).
%
\bibitem{chef2000}
Chefles,~A. 
Quantum state discrimination, 
\textit{Contemporary Physics} 
\textbf{41}, 401 (2000).
%
\bibitem{barn20091}
Barnett,~S.~M. \& Croke,~S. 
Quantum state discrimination. 
\textit{Adv. Opt. Photon.} 
\textbf{1}, 238 (2009).
%
\bibitem{berg2010}
Bergou,~J.~A. 
Discrimination of quantum states. 
\textit{J. Mod. Opt.} 
\textbf{57}, 160 (2010).
%
\bibitem{bae2015}
Bae,~J. \& Kwek,~L.-C. 
Quantum state discrimination and its applications. \textit{J. Phys. A: Math. Theor.} 
\textbf{48}, 083001 (2015).
%
\bibitem{benn19991}    
Bennett,~C.~H., DiVincenzo,~D.~P., Fuchs,~C.~A., Mor,~T., Rains,~E., Shor,~P.~W., Smolin,~J.~A. \& Wootters,~W.~K. 
Quantum nonlocality without entanglement. \textit{Phys. Rev. A} 
\textbf{59}, 1070 (1999).
%
\bibitem{pere1991}  
Peres~A. \& Wootters,~W.~K. 
Optimal detection of quantum information. \textit{Phys. Rev. Lett.} 
\textbf{66}, 1119 (1991).
%
\bibitem{duan2007}  
Duan,~R., Feng,~Y., Ji,~Z. \& Ying,~M. Distinguishing arbitrary multipartite basis unambiguously using local operations and classical communication.
\textit{Phys. Rev. Lett.} 
\textbf{98}, 230502 (2007).
%
\bibitem{chit2013}
Chitambar,~E. \& Hsieh,~M.-H. 
Revisiting the optimal detection of quantum information.
\textit{Phys. Rev. A} 
\textbf{88}, 020302(R) (2013).
%
\bibitem{ghos2001} 
S.~Ghosh, G.~Kar, A.~Roy, Sen(De),~A. \& Sen,~U.  
Distinguishability of Bell states. 
\textit{Phys. Rev. Lett.} 
\textbf{87}, 277902 (2001).
%
\bibitem{walg2002} 
Walgate,~J. \& Hardy,~L. 
Nonlocality, asymmetry, and distinguishing bipartite states.
\textit{Phys. Rev. Lett.} 
\textbf{89}, 147901 (2002).
%
\bibitem{fan2004} 
Fan,~H. 
Distinguishability and indistinguishability by local operations and classical communication. 
\textit{Phys. Rev. Lett.} 
\textbf{92}, 177905 (2004).
%
\bibitem{duan2009} 
Duan,~R., Feng,~Y., Xin,~Y. \& Ying,~M. 
Distinguishability of quantum states by separable operations.
\textit{IEEE Trans. Inf. Theory} 
\textbf{55}, 1320 (2009).
%
\bibitem{chit20141} 
Chitambar,~E., Duan,~R. \& Hsieh,~M.-H. 
When do local operations and classical communication suffice for two-qubit state discrimination?.
\textit{IEEE Trans. Inf. Theory} 
\textbf{60}, 1549 (2014).
%
\bibitem{band2015}
Bandyopadhyay,~S., Cosentino,~A., Johnston,~N., Russo,~V., Watrous,~J. \& Yu,~N. 
Limitations on separable measurements by convex optimization.
\textit{IEEE Trans. Inf. Theory} 
\textbf{61}, 3593 (2015).
%
\bibitem{band2021}  
Bandyopadhyay,~S. \& Russo,~V. 
Entanglement cost of discriminating noisy Bell states by local operations and classical communication.
\textit{Phys. Rev. A} 
\textbf{104}, 032429 (2021).
%
\bibitem{zhan2020}  
Zhang,~J.-H., Zhang,~F.-L., Wang,~Z.-X., Lai,~L.-M. \& Fei,~S.-M. 
Discriminating bipartite mixed states by local operations. 
\textit{Phys. Rev. A} 
\textbf{101}, 032316 (2020).
%
\bibitem{ha2022}  
Ha,~D. \& Kim,~J.~S. 
Bound on local minimum-error discrimination of bipartite quantum states.
\textit{Phys. Rev. A} 
\textbf{105}, 032421 (2022).
%
\bibitem{ivan1987} 
Ivanovic,~I.~D. 
How to differentiate between non-orthogonal states. 
\textit{Phys. Lett. A} 
\textbf{123}, 257 (1987).
%
\bibitem{pere1988} 
Peres,~A. 
How to differentiate between non-orthogonal states. 
\textit{Phys. Lett. A} 
\textbf{128}, 19 (1988).
%
\bibitem{diek1988} 
Dieks,~D. 
Overlap and distinguishability of quantum states. \textit{Phys. Lett. A} 
\textbf{126}, 303 (1988).
%
\bibitem{zhan2022}
Zhang,~J.-H., Zhang,~F.-L., Wang,~Z.-X. Yang, H. \& Fei, S.-M.
Unambiguous State Discrimination with Intrinsic Coherence.
\textit{Entropy} \textbf{24}, 18 (2022).
%
\bibitem{elda2004}  
Eldar,~Y.~C., Stojnic,~M. \& Hassibi,~B. 
Optimal quantum detectors for unambiguous detection of mixed states. 
\textit{Phys. Rev. A} 
\textbf{69}, 062318 (2004).
%
\bibitem{chef1998}  
Chefles,~A. \& Barnett,~S.~M.
Strategies for discriminating between non-orthogonal quantum states. 
\textit{J. Mod. Opt.} 
\textbf{45}, 1295-1302 (1998).
%
\bibitem{zhan1999}  
Zhang,~C.-W., Li,~C.-F. \& Guo,~G.-C.
General strategies for discrimination of quantum states. 
\textit{Phys. Lett. A} 
\textbf{261}, 25-29 (1999).
%
\bibitem{fiur2003}
Fiur\'a\ifmmode \check{s}\else \v{s}\fi{}ek,~J. \&  Je\ifmmode \check{z}\else \v{z}\fi{}ek,~M.
Optimal discrimination of mixed quantum states involving inconclusive results. 
\textit{Phys. Rev. A} 
\textbf{67}, 012321 (2003).
%
\bibitem{baga2012}
Bagan,~E., Mu\~noz-Tapia,~R., Olivares-Renter\'{\i}a,~G.~A. \& Bergou,~J.~A.
Optimal discrimination of quantum states with a fixed rate of inconclusive outcomes. 
\textit{Phys. Rev. A} 
\textbf{86}, 040303 (2012).
%
\bibitem{herz2015}
Herzog,~U.
Optimal measurements for the discrimination of quantum states with a fixed rate of inconclusive results. 
\textit{Phys. Rev. A} 
\textbf{91}, 042338 (2015).
%
\bibitem{boyd2004} 
Boyd,~S. \& Vandenberghe,~L. 
\textit{Convex Optimization} 
(Cambridge University Press, Cambridge, 2004).
%
\bibitem{sht}  
When $A$ and $B$ are disjoint convex sets in a real vector space $V$ with an inner product $\langle\cdot,\cdot\rangle$,
there exist $x\in\mathbb{R}$ and $\vec{v}\in V$ such that $\vec{v}\neq\vec{0}$ and $\langle \vec{a},\vec{v}\rangle\leqslant x\leqslant\langle \vec{b},\vec{v}\rangle$  for all $\vec{a}\in A$ and all $\vec{b}\in B$.
%
\end{thebibliography}
\end{document}